\documentclass[11pt]{article}

\usepackage{graphicx}
\usepackage{comment}
\usepackage{booktabs}
\usepackage{amssymb,amsmath,amsthm}
\usepackage{hyperref}
\usepackage{cleveref}
\usepackage{float}
\usepackage{wrapfig}
\usepackage{listings}
\usepackage{xcolor}
\usepackage{nicefrac}
\usepackage{tabularx}
\usepackage{multirow}
\usepackage{array}
\usepackage{nicefrac}
\usepackage{rotating}
\usepackage{afterpage}
\usepackage{lmodern}
\usepackage[most]{tcolorbox}

\usepackage[left=3cm,right=3cm]{geometry}

\newtheorem{definition}{Definition}
\newtheorem{lemma}{Lemma}
\newtheorem{theorem}{Theorem}

\newcommand{\etienne}[1]{}

\definecolor{std}{RGB}{130,130,130}

\lstdefinestyle{py}{
  language=Python,
  basicstyle=\ttfamily\small,
  columns=fullflexible,
  frame=single,
  breaklines=true,
  showstringspaces=false,
  keywordstyle=\color{blue!60!black},
  commentstyle=\color{green!50!black},
  stringstyle=\color{orange!60!black},
}

\newcolumntype{Y}{>{\centering\arraybackslash}X}

\usepackage[
  backend=biber,
  style=numeric,
]{biblatex}

\makeatletter
\renewcommand{\maketitle}{%
  \begin{center}
    {\LARGE\bfseries \@title \par}
    \vspace{1em}
    {\large \@author \par}
  \end{center}
  \vspace{2em}
}
\makeatother

\title{$k$-server-bench: Automating Potential Discovery for \\ the $k$-Server Conjecture}
\author{Kirill Brilliantov, Etienne Bamas, Emmanuel Abb\'e \\ MDS Lab, EPFL}

\begin{document}
\maketitle

\begin{abstract}
\setlength{\parindent}{0pt}
\setlength{\parskip}{0.4em}

\noindent We introduce a code-based challenge for automated, open-ended mathematical discovery based on the $k$-server conjecture, a central open problem in competitive analysis. The task is to discover a potential function satisfying a large graph-structured system of simple linear inequalities. The resulting evaluation procedure is sound but incomplete: any violated inequality definitively refutes a candidate, whereas satisfying all inequalities does not by itself constitute a proof of the corresponding conjecture's special case. Nevertheless, a candidate that passes all constraints would be strong evidence toward a valid proof and, to the best of our knowledge, no currently known potential achieves this under our formulation in the open $k=4$ circle case. As such, a successful candidate would already be an interesting contribution to the $k$-server conjecture, and could become a substantial theoretical result when paired with a full proof.

Experiments on the resolved $k=3$ regime show that current agentic methods can solve nontrivial instances, and in the open $k=4$ regime they reduce the number of violations relative to existing potentials without fully resolving the task. Taken together, these results suggest that the task is challenging but plausibly within reach of current methods.

Beyond its relevance to the $k$-server community, where the developed tooling enables researchers to test new hypotheses and potentially improve on the current record, the task also serves as a useful \emph{benchmark} for developing code-based discovery agents. In particular, our $k=3$ results show that it mitigates important limitations of existing open-ended code-based benchmarks, including early saturation and the weak separation between naive random baselines and more sophisticated methods.

\href{https://github.com/kibrq/k-server-bench}{\texttt{Code: https://github.com/kibrq/k-server-bench}}

\end{abstract}

\section{Introduction}

Recent advances in large language models (LLMs) and agentic scaffolds have made \emph{automated science} a practical research direction rather than a slogan: contemporary systems can routinely write code, propose hypotheses, run experiments, and iterate on failures \parencite{lu2024aiscientistfullyautomated,zheng2025automationautonomysurveylarge}. In mathematics, early task-specific efforts \parencite{davies2021guiding,peifer2020learning,wagner2021constructions,lample2019deep,charton2024patternboostconstructionsmathematicslittle,brilliantov2023applyinglanguagemodelsalgebraic} have recently been followed by more general-purpose agents that now can contribute beyond textbook problem solving, including program-search--driven discoveries and code-evolution workflows that improve sophisticated constructions \parencite{Romera-Paredes2024-cn,novikov2025alphaevolvecodingagentscientific,georgiev2025mathematicalexplorationdiscoveryscale,nagda2025reinforcedgenerationcombinatorialstructures,nagda2026reinforcedgenerationcombinatorialstructures}.

Inspired by these developments, we ask whether the same agentic capabilities can be brought to bear on \emph{open-ended discovery} in foundational theoretical computer science. We focus on the \emph{$k$-server conjecture}, widely regarded as the ``holy grail of competitive analysis'' and a flagship open problem in online algorithms. Since its formulation by Manasse, McGeoch, and Sleator in 1990 \cite{MANASSE1990208}, the conjecture has shaped the development of competitive analysis as a field \parencite{chrobak1991new,chrobak1991optimal,fiat1994competitive,Koutsoupias2009-jl,bartal2004competitive,bein20023}. It asks whether, for every parameter $k$ and every metric space, there exists a deterministic online algorithm whose total cost is at most $c=k$ times the offline optimum up to an additive constant. We refer to $c$ as the \emph{competitive ratio}, and call an algorithm satisfying this guarantee \emph{$c$-competitive}.

The best known general guarantee is due to \textcite{Koutsoupias1995-ro}, who proved that the \emph{Work Function Algorithm} (WFA) --- widely viewed as the leading candidate for resolving the conjecture --- is $(2k-1)$-competitive for any $k$ and any metric space. In many resolved metric families where WFA \emph{is} $k$-competitive, the proof proceeds by exhibiting a \emph{potential function} that satisfies a prescribed collection of inequalities \parencite{Koutsoupias2009-jl,Koutsoupias1995-ro}. Recent attention has concentrated on the \emph{circle} metric, a highly structured setting that nonetheless remains a major obstacle for existing potentials. \textcite{coester2021kserverconjectureunifyingpotential} introduced a unifying potential that resolves many known tight cases, yet it provably does not extend to the circle. By contrast, \textcite{huang2022deterministic3servercirclelimitation} settled the $k=3$ circle case by introducing a circle-specific modification, underscoring that the remaining difficulty appears to be discovering the right potential function.

Analysis of the development of this problem suggests that progress largely comes from identifying the right potential (or construction) and verifying it—a workflow naturally amenable to automation. We operationalize this paradigm by representing candidate potentials as Python programs and posing an explicit, automatically verifiable search problem: minimize the number of violated inequalities, with the natural target of \emph{zero} violations. A central roadblock is that the most interesting metric spaces are large or even infinite (such as the circle), making the full inequality system intractable to enumerate. To mitigate this, we evaluate potentials on a sequence of increasingly fine discretizations and on so-called $k$-taxi augmentations \parencite{coester2021kserverconjectureunifyingpotential} that capture additional structure of the infinite case. As a result, evaluation is sound but \textit{incomplete}: any violated inequality definitively refutes a candidate, whereas a candidate that satisfies all checked inequalities becomes a promising proof candidate for the infinite case. In practice, however, satisfying even this reduced constraint set is challenging; to the best of our knowledge, no potential currently satisfies all of the inequalities we enumerate for the $k=4$ setting, so discovering one would already constitute a substantial contribution.

We intentionally aim for a task that is challenging yet plausibly within reach of iterative LLM-based methods, rather than expecting models to produce fully formal proofs for arbitrary major conjectures. Three features support this: (i) the evaluator yields dense, automatic feedback enabling rapid iteration; (ii) the $k=3$ circle case is solved, providing a concrete calibration target; and (iii) empirically, agents can reach zero detected violations on $k=3$ instances and substantially reduce violations on the hardest $k=4$ instances beyond prior human-designed potentials in our evaluation suite. In this sense, our setup also contrasts with open-ended \textit{proof} benchmarks such as First Proof \parencite{abouzaid2026proof} and Erd\"os Bench \parencite{feng2026semiautonomousmathematicsdiscoverygemini,feng2026autonomousmathematicsresearch}, where automated feedback is typically sparser and iteration is less directly supported.

Moreover, our task can be viewed as a valuable \textit{benchmark} for agentic open-ended discovery. By providing a clear target---achieving \emph{zero} violations---it supports diagnostic evaluation beyond incremental score improvements \etienne{not sure about this sentence}. Even in the resolved $k=3$ regime, this criterion cleanly ranks methods from naive random sampling to more structured, agentic methods. This directly addresses a limitation of recent code-based open-ended discovery benchmarks \parencite{Romera-Paredes2024-cn,novikov2025alphaevolvecodingagentscientific,georgiev2025mathematicalexplorationdiscoveryscale}: their objectives often saturate and the progress is visible only through tiny decimal gains \parencite{novikov2025alphaevolvecodingagentscientific,lange2025shinkaevolveopenendedsampleefficientprogram,wan2025loongflowdirectedevolutionarysearch,wang2025thetaevolvetesttimelearningopen}, while random baselines can be surprisingly strong and obscure algorithmic contributions \parencite{gideoni2025random}. In contrast, our setup can track progress via the \emph{rate of perfect solutions} (e.g., the probability of reaching zero violations), not just marginal improvements in average score.

\paragraph{Contributions.}
\begin{enumerate}
    \item \textbf{An executable benchmark for an open problem.} We introduce an executable environment built around the $k$-server conjecture on the circle and cast potential-function proof search as a concrete optimization task with automatically checkable constraints and a clear target of \emph{zero} violations. We provide multiple benchmark instantiations (varying $k$, discretization, and optional $k$-taxi augmentation), along with a representation format and evaluation tooling that support rapid iteration and reproducible comparison across methods.
    \item \textbf{Initial baselines and agentic attempts.} We report initial experiments with several recent agent methods, as well as coding-agent plus human-in-the-loop workflows, demonstrating reliable solvability in the resolved $k=3$ regime and partial progress toward the more challenging open $k=4$ setting, while highlighting current limitations.
\end{enumerate}

Section \ref{sec: k-server-bench} formalizes the potential-based certification task for WFA as a collection of automatically checkable inequalities over an explicit \emph{work-function graph}, and describes how we instantiate this framework on discretized circle metrics and their $k$-taxi augmentations. Section \ref{sec: experiments} then presents the experimental baselines and agentic methods, including a detailed study of task-design choices and ablations in the resolved $k=3$ testbed, followed by results and failure modes in the more challenging $k=4$ regime. We conclude with a discussion of limitations (\autoref{sec: Limitations}) and conclusions (\autoref{sec: Conclusions}), and defer related work to \autoref{sec: Related Work}.

\section{$k$-server-bench}
\label{sec: k-server-bench}

We begin with a practically oriented exposition of our $k$-server-bench, deferring the mathematical background and additional details to the \autoref{sec: Mathematical Background on Work-Function Graph}.

\paragraph{Problem Formulation.}
Given a \textit{finite} metric space $M$, represented by a distance matrix in $\mathbb{R}^{|M|\times |M|}$, and an integer $k$, we construct a finite directed graph $(V,E)$, commonly referred to as the \emph{work-function graph}. Each node $u\in V$ is associated with a vector $w_u\in\mathbb{R}^{|\mathcal{C}|}$ (for simplicity we define $w$ as a vector but in infinite case it is a function from \emph{configurations}, $\mathcal{C} \mapsto \mathbb{R}$). Each directed edge is parameterized by a triple $(u,r,v)$: $u$ is the start vertex, $r$ is an edge parameter, and the end vertex $v$ is uniquely determined by the pair $(u,r)$. Edges carry weights, denoted by $\nabla(u,r, c)$, where $c$ is the competitiveness factor we want to prove (default is $k$).

The objective is to find a \emph{potential function} $\Phi:\mathbb{R}^{|\mathcal{C}|}\to\mathbb{R}$ such that the following inequality holds for every edge:
\[
    \forall (u,r,v)\in E:\qquad \Phi(w_v)-\Phi(w_u)\;\ge\;\nabla(u,r, c).
\]
If all inequalities are satisfied, then $\Phi$ certifies $c$-competitiveness of the Work Function Algorithm (WFA) on $M$.

To the best of our knowledge, the explicit notion of a \textit{work-function graph} has not been formalized in the $k$-server literature. Rather, closely related graph-based viewpoints appear as folklore and have likely been used implicitly in earlier analyses \parencite{chrobak1991optimal,coester2021kserverconjectureunifyingpotential}.

\paragraph{Family of Metric Spaces.}
As suggested by the formulation above, the task is closely related to detecting negative cycles in a directed graph. Indeed, running the Bellman--Ford algorithm on $(V,E)$ yields a certificate in the form of node potentials $\Phi$, but this certificate is purely \textit{numerical} and tied to a single finite instance.

Our goal, in contrast, is to obtain a \textit{functional} potential: a parameterized form that generalizes across multiple metric instances drawn from the \textit{same family} and can plausibly be integrated into a subsequent proof. We illustrate this with the circle metric. Ultimately, we seek a certificate for WFA on the \textit{infinite} circle, but the associated work-function graph is infinite, making direct verification infeasible. Instead, we evaluate candidates on a sequence of finite discretizations: for each $m$, we place $m$ equidistant points on the circle to obtain a finite metric, with values of $m$ accessible to modern computation (e.g., $m=4,5,6,\dots$). The desired potential should apply \emph{without modification} across these discretized instances. Both empirically and mathematically, success on an expanding range of $m$ provides strong evidence that the same functional form extends to the infinite-circle limit.

To clarify the distinction between an \emph{implicit} Bellman--Ford certificate and a \emph{functional} potential, consider the simple example $\Phi(w) = \sum_{X \in \mathcal{C}} w(X)$.
If such a potential happened to satisfy the required inequalities, it could in principle be incorporated into an infinite-case argument, because it is given as an explicit function of the work function $w$. By contrast, the Bellman--Ford potential is merely a numerical assignment to the nodes of a particular finite graph---that is, a lookup table for one discretized instance. As such, it does not directly provide a functional expression that could be transferred to the infinite setting.

This evaluation is sound but incomplete: any detected violation definitively refutes a candidate, whereas a candidate with no detected violations is best viewed as a strong hypothesis rather than a complete infinite-circle proof. Even so, such a candidate would be of independent interest---especially because, to our knowledge, no potential is currently known that achieves zero violations for the $k=4$ circle case in our discrete formulation. Finally, across all instances we can currently probe, we find no evidence that the resulting inequality systems are infeasible. Concretely, we use Bellman--Ford--style negative-cycle checks, which require at least one full pass over the edge set and thus become prohibitively expensive for the largest instance we consider ($k=4, m=8$ with $k$-taxi augmentation), whose work-function graph has on the order of billions of edges.

\begin{wraptable}{r}{0.55\textwidth}
  \centering
  \caption{Summary of the pre-generated circle discretizations. Extended discussion: Appendix \ref{subsec: Work Function Graph: Implementation Details}.}
  \label{tab:benchmark size}
  \vspace{0.5\baselineskip}
  \begin{tabular}{c|c|c|c|c}
    \toprule
    $k$ & $m$ & $k$-taxi & edges & violations \\ 
    \midrule
    $3$ & $6$ & no & $2'100$ & $0$ \\
    $3$ & $8$ & no & $41'920$ & $0$ \\
    $3$ & $6$ & yes & $193'662$ & $0$\\
    $3$ & $8$ & yes & $20'068'416$ & $1$ \\
    \midrule
    $4$ & $6$ & no & $6'006$ & $0$ \\
    $4$ & $8$ & no & $261'200$ & $0$ \\
    $4$ & $6$ & yes & $7'000'602$ & $17$ \\
    $4$ & $8$ & yes & $7 \cdot 10^9$$^*$ & $30'000$$^*$ \\
    \bottomrule
  \end{tabular}
  
\end{wraptable}

\paragraph{Benchmark Instances.}
From a practical standpoint, for the family of circle metrics we provide several pre-generated instances (see \autoref{tab:benchmark size}). In addition, we employ a convenient augmentation mechanism introduced by \textcite{coester2021kserverconjectureunifyingpotential}, known as \emph{$k$-taxi requests}. Intuitively, $k$-taxi requests allow us to substantially enrich the work-function graph over the same underlying set of sampled circle points: we add additional nodes and edges that approximate carefully constructed (infinite) request sequences. This augmentation both enlarges the constraint set and makes the potential inequalities more stringent, thereby increasing the likelihood that a potential satisfying all constraints will extend to the infinite-circle setting.

\autoref{tab:benchmark size} also reports the number of violations incurred by the unifying potential of \textcite{coester2021kserverconjectureunifyingpotential}. That work pushed the frontier of \emph{solved} metric spaces toward the circle by unifying several potentials originally developed for special cases, including the line metric, weighted star metrics, and the case $k=2$. To the best of our knowledge this unifying potential remains the state of the art among human-proposed candidates for the case $k=4$.

Note that, the pipeline is not specific to the circle. It can be adapted to other metric families with minimal changes, essentially by replacing the metric-defining parameter matrix while keeping the rest of the generation and evaluation framework unchanged.

\paragraph{Submission and Scoring.}
Submissions are provided as Python code defining a function \texttt{potential} that takes as input $k$, a metric $M$ (represented as a real-valued matrix), and a \emph{work function} $w$ (a vector associated with each node), and returns a single real value. See Appendix~\ref{subsec: Implementation Details: Potential Definition} for the complete specification.

Given a submission, we evaluate $\Phi$ by computing potential values for all nodes of the given instances and then check all inequalities. The primary benchmark target is the \emph{number of violated inequalities} incurred by the submitted potential (lower is better, with the ideal target being zero violations). Or alternatively, we use violations \textit{score}: $1 - \frac{\text{violations}}{\text{edges}}$ (higher is better).

In addition to the headline score, we report several diagnostic metrics to facilitate analysis and debugging, including the correlation between the submitted potential and the instance-specific Bellman--Ford potential, violation counts stratified by the type of edge weight $\nabla(u,r,c)$ and several more. Full definitions and reporting details are provided in the Appendix \ref{subsec: Metrics Definition}.

\section{Experiments}
\label{sec: experiments}

\paragraph{Methods.}

We compare three solution-generation strategies: naive Best-of-$N$ sampling, \textsc{ShinkaEvolve} \parencite{lange2025shinkaevolveopenendedsampleefficientprogram}, and \textsc{LoongFlow} \parencite{wan2025loongflowdirectedevolutionarysearch}. 

\textbf{Best-of-$N$.} Given a fixed prompt, we sample $N$ independent candidate implementations of \texttt{potential} from an LLM, evaluate each candidate on the benchmark instances, and return the single candidate with the lowest number of violations. This baseline captures the gains achievable purely from increased inference-time sampling without any iterative refinement loop.

\textbf{Evolutionary search paradigm.} Both \textsc{ShinkaEvolve} and \textsc{LoongFlow} follow the standard LLM-driven program-evolution template: maintain an archive/population of previously evaluated candidates; propose new candidates by mutating, recombining, or otherwise editing existing programs (with the LLM acting as a structured mutation operator); evaluate candidates with an automated fitness signal (here, violation count and related diagnostics); and iterate via selection. See \autoref{sec: Method Description and Used Hyperparameters} for the details on hyperparameters and discussion on the method differences.

\begin{table}[!t]
    \centering
    \caption{
Perfect runs and final scores summary (higher is better) over 10 independent runs of Best-of-$N$, \textsc{Shinka}, and \textsc{LoongFlow} all using \texttt{gpt-oss-120b} capped at 100 evaluations, evaluated on two instances ($k=3$) with $m\in{6,8}$ and no $k$-taxi augmentations. Methods are given with the same initial solutions and task descriptions.  ``Perfect Runs'' reports the fraction of runs achieving perfect score for each corresponding metric. ``Final Score'' reports per-run maxima for the combined metric (multiplicative across instances) and each instance, aggregated as mean $\pm$ \textcolor{std}{std} across runs.
}
    \label{tab: k=3 generalization search}
    \vspace{3pt}

\begin{tabularx}{\textwidth}{>{\raggedright\arraybackslash}Xcccccc}
\toprule
\multirow{2}{*}{Method} & \multicolumn{3}{c}{Perfect Runs} & \multicolumn{3}{c}{Max Score} \\
\cmidrule(lr){2-4}\cmidrule(lr){5-7}
 & combined & $m=6$ & $m=8$ & combined & $m=6$ & $m=8$ \\
\midrule
Best-of-N & 0/10 & 0/10 & 0/10 & $0.907 \textcolor{std}{\pm 0.064}$ & $0.958 \textcolor{std}{\pm 0.027}$ & $0.948 \textcolor{std}{\pm 0.039}$ \\
Shinka & 0/10 & 4/10 & 0/10 & $0.913 \textcolor{std}{\pm 0.146}$ & $0.956 \textcolor{std}{\pm 0.079}$ & $0.949 \textcolor{std}{\pm 0.090}$ \\
LoongFlow & 2/10 & 5/10 & 2/10 & $0.938 \textcolor{std}{\pm 0.084}$ & $0.972 \textcolor{std}{\pm 0.038}$ & $0.964 \textcolor{std}{\pm 0.052}$ \\
\bottomrule
\end{tabularx}

\end{table}

\subsection{Testbed: $k = 3$}

We begin by calibrating the benchmark and stress-testing key protocol choices in the resolved $k=3$ setting. As discussed earlier, $k=3$ on the circle is solved, and the corresponding benchmark instances are substantially smaller than those in the $k=4$ regime. This makes $k=3$ a natural entry point for controlled comparisons, while still reflecting the core difficulty we face at $k=4$: discovering potentials that generalize across discretizations rather than overfitting a single instance.

\paragraph{Design Choices.}
We adopt two protocol choices aimed at the open $k=4$ regime --- evaluating \emph{search procedures} rather single constructions and restricting to a \emph{canonical potential} hypothesis class --- and validate them in the $k=3$ testbed. See \autoref{sec: k=3 testbed experiments} for full information and supporting evidences on the following ablations.

\textbf{Search vs.\ construction.}
Across the testbed, we consistently found that, iterative improvement of the \textit{search procedures} over potentials outperforms direct iterative improvement of constructions (potentials in our case). Besides better scores, this formulation provides an additional scaling axis (the search can simply be run longer) and better matches what LLMs are trained to do: describe and adapt procedures. In our setting, that makes ``tuning a search algorithm'' a more natural use of the budget than fine-grained tuning of a single constructed potential.

\textbf{Canonical potentials.}
We find that current methods make little progress without a strong hint in the form of a substantially reduced hypothesis class: even for the smallest setting ($k=3, m=6$), unrestricted methods never discover a \textit{perfect} potential. Accordingly, we adopt the \emph{canonical potential} representation from \textcite{huang2022deterministic3servercirclelimitation}, in which potentials are parameterized by small integers $h,n$, an index matrix of size $h\times k$, and a coefficient vector of size $\binom{n}{2}$. This form is computationally efficient to evaluate and captures the structure of most of the known potentials (Appendix~\ref{subsec:canonical_potential}). Imposing this restriction improves performance dramatically and enables methods to reliably discover perfect potentials.

\paragraph{Stress test.}
After the ablations above, we run a $k=3$ generalization stress test designed to mimic the pressure of the $k=4$ regime: a single potential must satisfy the constraint suites for both $m=6$ and $m=8$. The best methods can still reach zero detected violations on the combined task, indicating that the proposed search-over-canonical-templates protocol admits feasible solutions. At the same time, \autoref{tab: k=3 generalization search} shows that aggregate scores (mean $\pm$ std.) overlap across methods, whereas the rate of zero-violation solutions (``Perfect Runs'') ranks them cleanly, supporting the usefulness of the proposed task as a \textit{benchmark}.

\subsection{Open case: $k = 4$}

We next turn to the open $k=4$ case on the circle with $k$-taxi augmentation. Here the goal is not a systematic comparison of methods, but to push beyond the best known human-designed baseline, namely the ``unifying potential'' of \textcite{coester2021kserverconjectureunifyingpotential}. Despite extensive search, we did not find a zero-violation potential at competitiveness $c=k$. The best candidate we obtained at $c=k$ incurs only $3$ violations out of roughly $7$ million checked inequalities on the $(k=4,m=6)$ augmented instance (cf.\ \autoref{tab:benchmark size}), while a separate candidate achieves zero violations on $(k=4,m=6)$ at the weaker ratio $c=k+1$ but fails to generalize to $m=8$. The summary of $k$-competitive candidates is given in the \autoref{tab: final leaderboard k = 4}.

\begin{wraptable}{r}{0.5\textwidth}
\centering
\caption{Best-found $k$-competitive potential. Candidates are evaluated on the $k$-taxi augmented $k=4,m=6$ circle metric.}
\label{tab: final leaderboard k = 4}
\vspace{4pt}

\begin{tabularx}{0.45\textwidth}{@{}l Y@{}}
\toprule
Found by & Violations \\
\midrule
Codex (gpt-5.2) & 3  \\
\textsc{Shinka} (gpt-oss-120b) & 14 \\
Humans \parencite{coester2021kserverconjectureunifyingpotential} & 17 \\
\bottomrule
\end{tabularx}
\end{wraptable}

Moving from $k=3$ to the open $k=4$ regime exposes several bottlenecks for evolutionary and agentic search. The baseline is already extremely close to feasibility (only $17$ violations out of $\approx 7$ million constraints), so progress requires operating in the ``tens of violations out of millions'' regime. At the same time, full evaluation is expensive, which sharply limits candidate throughput per iteration, and naive approximations such as evaluating on a random subset of constraints are ineffective at this resolution, since they do not reliably discriminate near-feasible candidates.

To make exploration tractable, we introduced an additional degree of freedom in the interface: agents are allowed to define and optimize a \emph{proxy} objective during search, leaving full evaluation only for the best-found candidates. This improves throughput but comes with practical tradeoffs: the evaluation contract becomes more complex, solutions and prompts grow, and token consumption increases, which raises cost for frontier models significantly and increases failure rates from brittle long-context code edits (details in the \autoref{sec: k = 4 scaling challenges}).

We also explored stronger restrictions on the hypothesis class beyond canonical potentials. In one regime, search was constrained to perturbations of the unifying potential. To construct a second structured hint, we performed an iterative GPT-5 Pro–assisted analysis aimed at suggestions on generalizing the $k=3$ circle potential of \textcite{huang2022deterministic3servercirclelimitation} to the $k=4$ case.

Overall, both with the canonical hint alone and with the additional structured hint, the best solutions plateaued at approximately 5{,}000 violations (see \autoref{sec: k = 4 scaling challenges} for details). When the canonical potential itself was revealed, all methods became stuck in a sharp local minimum around this value. Only a single run of \textsc{ShinkaEvolve} (using \texttt{gpt-oss-120b}) produced a modest improvement to 14 violations. However, this improvement resulted from a minor modification and did not lead to further progress, suggesting that this direction is unlikely to yield a complete solution.

\paragraph{Human-guided search.} However, the second hint proved most useful when paired with a coding-agent plus human-in-the-loop workflow, particularly Codex (powered by \texttt{gpt-5.2}). We gave Codex the same task description and the ability to evaluate candidate solutions autonomously, while human feedback during the loop was limited to conceptual guidance about the $k$-server domain rather than code-level edits. A key contribution of Codex in this process was the implementation of an effective proxy evaluation routine: it maintained a cache of ``hard'' edges that were frequently violated and used early stopping to discard weak candidates quickly. This substantially increased candidate throughput by allowing the search to screen many proposals cheaply before running full verification. Leveraging this proxy, the workflow produced the current best $c=k$ candidate (3 violations on the augmented $k=4,m=6$ instance) and also uncovered a $c=k+1$ potential with zero violations on $k=4,m=6$, though it does not generalize to the augmented $m=8$ setting. We defer the discovered potential structures to \autoref{sec: Potential, Codex with Human-in-the-Loop} and additional details on to the \autoref{sec: Codex-found Solution}.

\section{Limitations}
\label{sec: Limitations}

Our benchmark inherits several fundamental limitations from the underlying proof-search problem. First, instance size grows rapidly with $k$, discretization level $m$, and augmentations, leading to very rapid growth in the number of states and constraints that must be checked. This makes exhaustive evaluation increasingly expensive and forces the use of proxy objectives or partial checking at larger scales. Second, while the circle (and its $k$-taxi augmentation) provides a particularly crisp and challenging stress test---including a strong human baseline that is close to feasibility---there are relatively few metric families with similarly well-motivated, discretizable open cases and comparable ``near-miss'' potentials. Third, our representation assumes candidate potentials are efficiently computable programs; however, a correct potential might require structure that is difficult to express or evaluate efficiently in Python under our current interface. Finally, although canonical parameterizations are computationally convenient, existing results suggest that broad canonical classes have intrinsic limitations \parencite{huang2022deterministic3servercirclelimitation}, and over-reliance on them may bias search away from genuinely new potential forms.

\section{Conclusions}
\label{sec: Conclusions}

We introduced an executable environment that operationalizes potential-function proof search for the $k$-server conjecture on the circle as a code-driven optimization task with automatically checkable constraints and a clear target of zero violations. Across experiments, we find that dense automated feedback enables rapid iteration and meaningful progress, including improvements beyond strong human-designed baselines in the open $k=4$ regime under our formulation. More broadly, our results suggest that, before attempting fully autonomous long-horizon runs, it can be valuable to \emph{calibrate} the benchmark via interactive sessions with coding agents and human-in-the-loop workflows, both to refine the interface (e.g., evaluation/proxy design) and to identify effective hypothesis classes.

We hope that releasing this evaluation suite will accelerate progress toward a proof in the open $k=4$ circle case, and more broadly encourage the development of similar executable benchmarks for other open mathematical problems that admit automatic checking. Our broader aim is to foster settings in which progress can be measured reliably, hypotheses can be tested at scale, and modern AI systems can be tightly integrated with classical mathematical workflows.

\printbibliography

\appendix

\section{Related Work}
\label{sec: Related Work}

Recent AI systems have shown growing promise in formal scientific domains, including both physics \parencite{brenner2026solvingopenproblemtheoretical,guevara2026singleminusgluontreeamplitudes} and mathematics \parencite{georgiev2025mathematicalexplorationdiscoveryscale,dobriban2025solvingresearchproblemmathematical,bubeck2025earlyscienceaccelerationexperiments,feng2026semiautonomousmathematicsdiscoverygemini,feng2026autonomousmathematicsresearch,sothanaphan2026resolutionerdhosproblem728,chen2026felsconjecturesyzygiesnumerical,chen2026paritykdifferentialsgenuszero}. Alongside these application-oriented results, there is a substantial line of work on AI systems and search algorithms for mathematical and scientific discovery \parencite{Romera-Paredes2024-cn,novikov2025alphaevolvecodingagentscientific,openevolve,lange2025shinkaevolveopenendedsampleefficientprogram,surina2025algorithmdiscoveryllmsevolutionary,jiang2026deltaevolveacceleratingscientificdiscovery,wan2025loongflowdirectedevolutionarysearch}, as well as on applying such systems to concrete open mathematical tasks \parencite{novikov2025alphaevolvecodingagentscientific,georgiev2025mathematicalexplorationdiscoveryscale,nagda2025reinforcedgenerationcombinatorialstructures,nagda2026reinforcedgenerationcombinatorialstructures}. A large fraction of these works focus on open-ended construction tasks, such as improving upper or lower bounds, where progress is measured by small numerical improvements. While this paradigm is valuable, it can be difficult to interpret scientifically and often yields benchmarks that do not clearly distinguish methods. In contrast, our setting has a sharp and meaningful target---zero violations on the generated instances---together with a direct connection to a longstanding open problem in competitive analysis. In particular, success on our benchmark would constitute progress toward the $k$-server conjecture, and the closely related $k=3$ circle case has already led to a substantial theoretical result published in a reputable venue \parencite{huang2022deterministic3servercirclelimitation}.

A second line of work develops benchmark suites built directly around research-level mathematical problems \parencite{glazer2025frontiermathbenchmarkevaluatingadvanced,epochFrontierMathOpenProblems,abouzaid2026proof,feng2026semiautonomousmathematicsdiscoverygemini,feng2026autonomousmathematicsresearch}. FrontierMath \parencite{glazer2025frontiermathbenchmarkevaluatingadvanced} focuses primarily on difficult mathematical problems with known answers, while its more recent Open Problems track \parencite{epochFrontierMathOpenProblems} extends this paradigm to open questions with verifiable final answers. These benchmarks are philosophically close to our design in that they aim for objective verification, but they also have important limitations for agent research: parts of the evaluation stack remain closed or paywalled, which restricts independent experimentation, and some verifiers can certify correctness only with high confidence rather than absolute certainty. Our benchmark shares the general challenge that verification may be incomplete in the infinite setting, but differs in being fully executable and openly inspectable on finite instances. Related proof-oriented benchmarks, such as First Proof \parencite{abouzaid2026proof} and Erd\H{o}s Bench \parencite{feng2026autonomousmathematicsresearch,feng2026semiautonomousmathematicsdiscoverygemini}, evaluate the ability of models to produce complete mathematical proofs. Compared with these settings, our task is more interactive: agents receive dense, instance-specific feedback, can iteratively improve candidate constructions, and do not require human judgment in the evaluation loop. This makes large-scale automated experimentation substantially easier and, we hope, more conducive to rapid progress on the underlying mathematical problem.

Finally, within competitive analysis itself, there has been recent work using AI to study online algorithms, for example in online matching where reinforcement learning is used to tighten competitive bounds \parencite{pmlr-v235-zhang24bf}. Our setting is qualitatively different. In online matching, the emphasis is often on finding improved algorithms or counterexamples, whereas for the $k$-server conjecture there is already a canonical candidate algorithm---the Work Function Algorithm---and the central challenge is to prove its competitiveness. Accordingly, our benchmark is aimed not at algorithm discovery, but at discovering proof-relevant potential functions that could ultimately contribute to resolving a central conjecture in online algorithms.

\section{Mathematical Background on Work-Function Graph}
\label{sec: Mathematical Background on Work-Function Graph}

\subsection{Work Function}

We call a multiset of $k$ points in the metric space $M$ a configuration (of servers) and denote $\mathcal{C}$ with a set of all possible configurations.

The central theoretical object in the $k$-server problem is the \emph{work function}.

\begin{definition}
The work function $w_t : \mathcal{C} \to \mathbb{R}$ maps a configuration $C$ to the minimum cost of serving the request sequence $r_1,\dots,r_t$ and ending in configuration $C$.
\end{definition}

It satisfies the recurrence
\[
    w_{t+1}(X)
    = T_r(w_t)(X) =
    \min_{\substack{Y \in \mathcal{C}\\ r_{t+1}\in Y}}
    \left\{
        w_t(Y) + d(X,Y)
    \right\},
\]
where the minimum is taken over all configurations $Y$ that contain $r_{t+1}$ and $T_r$ is the convenient notation for operator returning the ``next'' work function. Here, with slight abuse of notation, $d(X,Y)$ denotes the minimum matching cost between configurations $X$ and $Y$.

The work function enjoys several important structural properties, including monotonicity, quasi-convexity, and $1$-Lipschitzness \parencite{Koutsoupias2009-jl}.

\begin{definition}[Extended Cost]
Let $w$ be a work function. The \emph{extended cost} is defined as
\[
    \nabla(w,r)
    :=
    \max_{X \in \mathcal{C}}
    \bigl(
        T_r(w)(X) - w(X)
    \bigr).
\]
\end{definition}

\begin{definition}[OPT Increase]
Let $w$ be a work function. The \emph{OPT increase} is defined as
\[
    \Delta \operatorname{OPT}(w,r)
    :=
    \min_{X \in \mathcal{C}} T_r(w)(X)
    -
    \min_{X \in \mathcal{C}} w(X).
\]
\end{definition}

\subsection{Competitiveness of the WFA}

The Work Function Algorithm (WFA) uses the work function explicitly to determine its moves. While WFA is defined algorithmically, the work function itself is a more general analytical object.

If the current configuration after $t$ requests is $X_t$, WFA chooses the next configuration according to
\[
    X_{t+1}
    =
    \operatorname*{arg\,min}_{Y \in \mathcal{C}}
    \left\{
        w_{t+1}(Y) + d(X_t,Y)
    \right\}.
\]

\begin{lemma}
\label{lemma:main_lemma}
For any request sequence $\rho = (r_1,\dots,r_t)$,
\[
    \operatorname{OPT}(\rho)
    +
    \operatorname{WFA}(\rho)
    \le
    \sum_{i=0}^{t - 1}
    \nabla(w_i,r_{i+1}).
\]
\end{lemma}

A standard way to prove $c$-competitiveness of WFA is to construct a potential $\Phi$ on work functions.

\begin{theorem}
\label{theorem:main_theorem}
Suppose there exists $\Phi$ such that:
\begin{enumerate}
    \item For all $w$ and requests $r$,
    \[
        \Phi(w') - \Phi(w)
        \ge
        \nabla(w,r),
    \]
    where $w' = T_r(w)$ is the work function after $r$,
    \item For all $w$,
    \[
        \Phi(w)
        \le
        (c+1)\min_{X\in\mathcal{C}} w(X)
        + \mathrm{const}.
    \]
\end{enumerate}
Then WFA is $c$-competitive.
\end{theorem}

\begin{proof}
Using Lemma~\ref{lemma:main_lemma} and telescoping,
\[
\operatorname{OPT}(\rho)
+
\operatorname{WFA}(\rho)
\le
\sum_i \nabla(w_i,r_{i+1})
\le
\Phi(w_{t+1}) - \Phi(w_0),
\]
and the second property yields the claim after absorbing constants.
\end{proof}

\subsection{Work-Function Graph}

We first note that work functions are shift invariant in the following sense:
\begin{lemma}[Shift-Invariance]
    For every $w: \mathcal{C} \mapsto \mathbb{R}$, request $r \in M$ and scalar $\alpha \in \mathbb{R}$ we have:
    \[
        T_r(w + \alpha \cdot \mathbf{1}) = T_r(w) + \alpha\cdot\mathbf{1},
    \]
    where $+\alpha\cdot\mathbf{1}$ means adding $\alpha$ to every value of the work function.

    Consequently,
    \[
        \nabla(w + \alpha\cdot 1, r) = \nabla(w, r), \quad \Delta\mathrm{OPT}(w + \alpha\cdot\mathbf{1}, r) = \Delta\mathrm{OPT}(w, r)
    \]
\end{lemma}

\begin{definition}[Normalized Work Function]
    For a work function $w$ we call \[\widehat{w} = w - \min\limits_{X\in\mathcal{C}} w(X) \cdot \mathbf{1}\] to be normalized work function.    
\end{definition}

\begin{definition}[Reachable state]
    Normalized work function $\widehat{w}$ is reachable from $w_0$ if there exists a finite sequence of requests $r_1, ..., r_n$ such that $\widehat{w} = \widehat{w_n}$, where $w_i = T_{r_i}(w_{i - 1})$
\end{definition}

\begin{lemma}[Finite number of Reachable states]
    For a finite metric space $M$ with integer (for simplicity) distances, the number of reachable normalized work functions from an integer-valued work function, $w_0$, is finite.
\end{lemma}
\begin{proof}
    First, we note that the application of $T_r$ operator keeps the work function to be integer valued since all distances are as well integers.

    Next we note that because of Lipschitzness of the work function all normalized work functions values in the given metric space are bounded:

    \[
        \forall w, X: \widehat{w}(X) \leq \max\limits_{X, Y \in \mathcal{C}} d(X, Y)
    \]

    Given that all values of the normalized work functions are bounded and integer valued we have that there are only finite number of reachable normalized work functions.
\end{proof}

\begin{definition}[Work Function Graph]
    The work function graph is the directed graph $G = (V, E)$ whose vertex set is the finite set $V$ of reachable normalized work functions and where for each $\widehat{w} \in V$ and request $r \in M$ we add the edge
    \[
        e = (\widehat{w_u}, r, \widehat{w_v}) \quad \text{whenever} \quad \widehat{w_v} = \widehat{T_r(\widehat{w_u})}
    \]
    For such an edge, define its weight
    \[
        \mathrm{weight}_c(u, r, v) = \nabla(w_u, r) - (c + 1) \Delta\mathrm{OPT}(w_u, r)
    \]
\end{definition}

\begin{lemma}[Normalized Potential Criterion]
    Suppose $\Psi: V \mapsto \mathbb{R}$ satisfies
    \[
        \Psi(\widehat{w_v}) - \Psi(\widehat{w_u}) \geq \mathrm{weight}_c(u, r, v)
    \]
    for every edge $e = (u, v, r)\in E$. Define
    \[
        \Phi(w) := \Psi(\widehat{w}) + (c + 1) \min\limits_{X\in\mathcal{C}} w(X)
    \]
    Then for every reachable work function $w$ and $w' = T_r(w)$,
    \[
        \Phi(w') - \Phi(w) \geq \nabla(w, r) \quad \text{and} \quad \Phi(w) \leq (c + 1) \min\limits_{X} w(X) + B,
    \]
    for some constant $B$ that does not depend on the request sequence to reach $w$.
\end{lemma}
\begin{proof}
    By definition,
    \[
        \Phi(w') - \Phi(w) = \Psi(\widehat{w'}) - \Psi(\widehat{w}) + (c + 1) \Delta\mathrm{OPT}(w, r)
    \]
    The assumed inequality gives:
    \[
        \Phi(w') - \Phi(w) \geq  \mathrm{weight}_c(u, r, v) + (c + 1) \Delta\mathrm{OPT}(w, r) = \nabla(w, r)
    \]

    Moreover, since $V$ is finite we can set $B = \max\limits_{v \in V} \Psi(\widehat{w_v})$, then since $w$ is reachable we have:
    \[
        \Phi(w) = \Psi(\widehat{w}) + (c + 1)\min\limits_{X \in \mathcal{C}} w(X) \leq (c + 1)\min\limits_{X \in \mathcal{C}} w(X) + B
    \]
\end{proof}

Therefore combining all these facts we get the finiteness of the BFS procedure to enumerate all possible normalized work functions, the sufficient of computing potential on the normalized work functions. So the candidate satisfying all inequalities in the Work Function Graph of the finite metric space with integer values constitutes the competitiveness of the Work Function Algorithm on this metric space.

\subsection{Implementation Details}
\label{subsec: Work Function Graph: Implementation Details}

Enumerating the largest instance, namely $k=4$, $m=8$ with $k$-taxi augmentation, via a straightforward BFS required prohibitive memory. To make enumeration feasible, we used Bloom filters to maintain the set of visited nodes. As a result, the counts reported in \autoref{tab:benchmark size} are not exact; they should be interpreted as lower bounds.

In addition, for instances with $k$-taxi requests we applied an extra node-filtering procedure exploiting the rotational and reflection symmetries of the circle, which further reduced the effective number of \emph{distinct} nodes. 

\section{Potential Forms}

\subsection{Canonical Potential}
\label{subsec:canonical_potential}

Let $n \in \mathbb{N}$, an index matrix $I \in \mathbb{Z}^{h \times k}$, and a symmetric coefficient matrix $C \in \mathbb{R}^{n \times n}$.  
We define the \emph{canonical potential} parameterized by auxiliary points $a_1,\dots,a_n \in M$ as
\[
    \Phi(w; a_1,\dots,a_n)
    =
    \sum_{i=1}^{h}
        w\!\bigl(a_{I[i,1]}, \dots, a_{I[i,k]}\bigr)
    -
    \sum_{1 \le i < j \le n}
        C[i,j]\, d(a_i,a_j),
\]
where:
\begin{itemize}
    \item $I[i,j]$ denotes the $(i,j)$-th entry of the index matrix,
    \item $w(a_{i_1},\dots,a_{i_k})$ denotes the value of the work function at configuration $(a_{i_1},\dots,a_{i_k})$,
    \item negative indices are interpreted via the convention $a_{-i} = \overline{a_i}$ (the antipode of $a_i$ on the circle).
\end{itemize}

The final potential is defined by minimizing over the auxiliary points:
\[
    \Phi(w)
    =
    \min_{a_1,\dots,a_n \in M}
    \Phi(w; a_1,\dots,a_n).
\]

Due to the $1$-Lipschitz property of the work function, any canonical potential of the above form satisfies the upper bound
\[
    \Phi(w)
    \;\le\;
    h \cdot \min_{X \in \mathcal{C}} w(X)
    + \mathrm{const},
\]
where $h$ is the number of rows in the index matrix $I$, and the constant depends only on the metric space and the chosen coefficient matrix $C$ (but not on the request sequence). Thus, the number of rows $h$ effectively controls the achievable competitiveness bound encoded by the potential.

We also note that this general framework was introduced in \textcite{huang2022deterministic3servercirclelimitation} in a slightly more restricted form. In that formulation, the index matrix is required to contain a row of the form $(-1,-1,\dots,-1)$, and in all other rows the index $1$ appears explicitly, with all distance coefficients involving point $1$ set to zero. Our definition relaxes these structural constraints while preserving the same underlying minimization template.

In our implementation, instead of storing the full symmetric matrix $C$, we use a coefficient vector of size $\binom{n}{2}$.

\paragraph{Examples from the Literature.}

The unifying potential of \textcite{coester2021kserverconjectureunifyingpotential} fits into this framework with
\[
    n := k,
    \qquad
    I :=
    \begin{pmatrix}
    1 & 2 & \dots & k \\
    -1 & 2 & \dots & k \\
    -2 & -2 & \dots & k \\
    \vdots & \vdots & \ddots & \vdots \\
    -k & -k & \dots & -k
    \end{pmatrix},
    \qquad
    C := 0.
\]

For $k=3$, the potential of \textcite{huang2022deterministic3servercirclelimitation} is given by
\[
    n := 4,
    \qquad
    I :=
    \begin{pmatrix}
    -1 & -1 & -1 \\
    1 & 2 & -3 \\
    1 & 3 & -4 \\
    1 & 4 & -2
    \end{pmatrix},
\qquad
    C :=
    \begin{pmatrix}
    0 & 0 & 0 & 0 \\
    0 & 0 & 1 & 1 \\
    0 & 1 & 0 & 1 \\
    0 & 1 & 1 & 0
    \end{pmatrix}.
\]

Equivalently, these potentials can be written in the original form:
\[
    \Phi(w) = \min\limits_{x_1, ..., x_k \in M}\sum\limits_{i=0}^{k} w(\overline{x_i}, ..., \overline{x_i}, x_{i+1}, ..., x_k)
\]
and
\[
\begin{aligned}
    \Phi(w) &= \min\limits_{u, x, y, z} w(\overline{u}, \overline{u}, \overline{u}) + w(u, x, \overline{y}) + w(u, y, \overline{z}) + w(u, z, \overline{x}) - \\ &
    - d(x, y) - d(x, z) - d(y, z)
\end{aligned}
\]
correspondingly

\subsection{Discovered Potentials}
\label{sec: Potential, Codex with Human-in-the-Loop}

Below we report representative canonical potentials discovered with human-in-the-loop assistance.

\paragraph{$k$-competitive candidate.}

\[
    n := 5,
\qquad
    I :=
    \begin{pmatrix}
    -5 & -5 & -5 & -5 \\
    5 & -1 & -2 & -2 \\
    5 & 1 & 3 & 4 \\
    5 & 2 & -4 & -4 \\
    5 & 2 & 4 & -3
    \end{pmatrix},
\qquad
    C :=
    \begin{pmatrix}
    0 & -1 & 0 & -1 & 0 \\
    -1 & 0 & 1 & 0 & 0 \\
    0 & 1 & 0 & -1 & 0 \\
    -1 & 0 & -1 & 0 & 0 \\
    0 & 0 & 0 & 0 & 0
    \end{pmatrix}.
\]

\paragraph{$(k+1)$-competitive candidate.}

\[
    n := 4,
    \qquad
    I :=
    \begin{pmatrix}
    1 & 2 & 3 & 4 \\
    1 & 2 & 3 & 4 \\
    -1 & 2 & 3 & 4 \\
    -2 & -2 & 3 & 4 \\
    -3 & -3 & -3 & 4 \\
    -4 & -4 & -4 & -4
    \end{pmatrix},
    \qquad
    C := 0.
\]

This construction closely resembles the unifying potential for $k=4$. The main difference is that one work-function term (e.g., $(1,2,3,4)$ or equivalently $(-1,2,3,4)$ up to symmetry) appears with multiplicity two.

\paragraph{Shinka's modification of Unifying Potential.}
The index matrix is the same as in the Unifying Potential case, but it increased the $n$ from $4$ to $5$ by effectively introducing additional point to the minizmiation and connecting it to the rest of the points only via distance coefficients:

\[
    C := \begin{pmatrix}
        0 & 1 & -1 & -1 & -1 \\
        1 & 0 & -1 & -1 & -1 \\
        -1 & -1 & 0 & 1 & 1 \\
        -1 & -1 & 1 & 0 & 1 \\
        -1 & -1 & 1 & 1 & 0
    \end{pmatrix}
\]
So in the non-matrix form:
\[\begin{aligned}
\Phi(w) &= \min\limits_{a_1,\dots,a_5}
\Bigl(
  w(a_1,a_2,a_3,a_4)
  + w(\overline{a_1},a_2,a_3,a_4)
  + w(\overline{a_2},\overline{a_2},a_3,a_4) \\
&\qquad\qquad
  + w(\overline{a_3},\overline{a_3},\overline{a_3},a_4)
  + w(\overline{a_4},\overline{a_4},\overline{a_4},\overline{a_4})
\Bigr) \\
&\quad - \Bigl(
  d(a_1,a_2) - d(a_1,a_3) - d(a_1,a_4) - d(a_1,a_5)
  - d(a_2,a_3) - d(a_2,a_4) - d(a_2,a_5) \\
&\qquad\qquad
  + d(a_3,a_4) + d(a_3,a_5) + d(a_4,a_5)
\Bigr).
\end{aligned}\]

\section{Evaluation Framework Implementation Details}
\label{sec: Evaluation Framework Implementation Details}
\subsection{Potential Definition}
\label{subsec: Implementation Details: Potential Definition}

Our reference implementation expects submissions to define a \texttt{Potential} object with a lightweight callable interface:

\begin{lstlisting}[style=py]
class Potential:
    def __init__(self, context: WFContext, **kwargs):
        pass

    def __call__(self, wf: np.ndarray) -> float:
        pass
\end{lstlisting}

Here, \texttt{WFContext} is a helper class that encapsulates the metric instance (via the distance matrix) and provides utilities for working with work functions and configurations:

\begin{lstlisting}[style=py]
class WFContext:
    def __init__(self, k: int, distance_matrix: np.ndarray):
        pass

    def config_to_idx(self, config: Tuple[int, ...]) -> int:
        """
        Maps `config` (a tuple of point indices) to its index in the work-function
        vector, i.e., w(C) = wf[idx] where config_to_idx(C) = idx.
        """
        pass

    def idx_to_config(self, idx: int) -> Tuple[int, ...]:
        """
        Inverse mapping of `config_to_idx`.
        """
        pass

    def update_wf(self, wf: np.ndarray, r: int) -> np.ndarray:
        """
        Returns the updated work function after serving request r.
        """
        pass

    ...
\end{lstlisting}

We support precomputation via \texttt{\_\_init\_\_} so that a submission can amortize expensive setup and keep \texttt{\_\_call\_\_} fast. This is particularly important because evaluation requires many repeated calls to the potential over large work-function graphs. For example, our \texttt{CanonicalPotential} implementation precomputes indices of relevant configurations:

\begin{lstlisting}[style=py]
class CanonicalPotential:
    def __init__(
        self, context, n: int,
        index_matrix: np.ndarray,
        coefs: np.ndarray
    ):
        # shape: [N_CANDIDATES x k]
        self.candidate_config_idxes = ...

        # shape: [N_CANDIDATES]
        self.penalties = ...

    def __call__(self, wf: np.ndarray) -> float:
        wf_part = wf[self.candidate_config_idxes].sum(axis=-1)
        return (wf_part - self.penalties).min()
\end{lstlisting}

\subsection{Metrics Definition}
\label{subsec: Metrics Definition}

\paragraph{Core Metrics.}

\paragraph{\texttt{violations\_k}.}  
Number of edges violating the normalized inequality
\[
\Phi(\hat{w_v}) - \Phi(\hat{w_u}) + (k+1)\,\Delta{\mathrm{OPT}}(w_u, r)
\;\ge\;
\nabla(w_u, r).
\]
Equivalently, this corresponds to the step inequality
\[
\Phi(w_{v}) - \Phi(w_u)
\;\ge\;
\nabla(w_u,r_{t+1})
\]
in the normalized work-function graph.  
This is the \emph{primary optimization target}: the smaller \texttt{violations\_k}, the closer the potential is to satisfying all required inequalities.  
The optimal potential achieves \texttt{violations\_k} $= 0$.

\paragraph{\texttt{violations\_k\_l\{1,2,inf\}}.}  
The $\ell_1$, $\ell_2$, and $\ell_\infty$ norms of violation magnitudes.  
These quantify the \emph{severity} of violations. For example, a potential may have many small violations (large $\ell_1$ but small $\ell_\infty$), suggesting it is close to feasibility and may be repairable via small structural changes.

\paragraph{\texttt{violations\_dmin\_0}.}  
Number of violated edges with $\Delta{\mathrm{OPT}} = 0$.  
For such edges the inequality reduces to
\[
\Phi(\hat{w_v}) - \Phi(\hat{w_u})
\;\ge\;
\nabla(w_u, r),
\]
with no contribution from the OPT increase term.  
These transitions are particularly demanding and measure how well the potential handles \emph{zero-OPT-gap} edges.

\paragraph{\texttt{detected\_dmin\_0\_score}.}  
Fraction of ``hard'' edges (those with $\Delta{\mathrm{OPT}} = 0$ and $\nabla(w_u, r) > 0$) for which $\Phi(w_v) \neq \Phi(w_u)$.  
This measures \emph{sensitivity}: whether the potential meaningfully distinguishes difficult transitions.  
For instance, a trivial constant potential may incur few violations under relaxed metrics, yet will typically fail to detect any such hard edges.

\paragraph{\texttt{violations\_renorm}.}  
Number of violations in the \emph{renormalized} graph, evaluated as
\[
\Phi\bigl(\hat{w_{v}} + \mathbf{1}\cdot\Delta{\mathrm{OPT}}\bigr)
-
\Phi\bigl(\hat{w_u}\bigr)
\;\ge\;
\nabla(w_u,r),
\]
where $+ \mathbf{1} \Delta\operatorname{OPT}$ means adding the same constant to all work function values.
This metric serves as a consistency check under normalization.

For example, a simple sum-type potential often performs very well under renormalized evaluation, since normalization removes scaling artifacts. However, such potentials typically perform poorly on \texttt{violations\_k}, as they implicitly correspond to a very large competitive ratio. Because \texttt{violations\_k} balances both feasibility and competitive ratio, a low \texttt{violations\_renorm} alone does \emph{not} indicate a strong potential.

\paragraph{Competitiveness Metrics.}

\paragraph{\texttt{strong\_hypothesis\_rho}.}  
The smallest value $\rho$ such that all edges with $\Delta{\mathrm{OPT}} > 0$ satisfy
\[
\Phi(w_v) - \Phi(w_u)
+
(\rho+1)\,\Delta{\mathrm{OPT}}(w_u, r)
\;\ge\;
\nabla(w_v, r).
\]
This estimates the \emph{tightest attainable competitive ratio} for the given potential structure.

\paragraph{\texttt{opt\_upper\_bound}.}  
A structural proxy for the competitive ratio; it estimates the growth rate of $\Phi(\cdot)$ with adding same constant value to all work function values.

\paragraph{Interpretation.}
\begin{itemize}
    \item If \texttt{opt\_upper\_bound} $< k+1$, then \texttt{violations\_k} cannot reach zero since every algorithm is at least $k$ competitive. 
    \item If \texttt{strong\_hypothesis\_rho} $< \texttt{opt\_upper\_bound}$ and \texttt{violations\_dmin\_0} $= 0$, the competitive ratio can potentially be reduced.
    \item If \texttt{strong\_hypothesis\_rho} $> \texttt{opt\_upper\_bound}$, the targeted competitive ratio must be increased.
\end{itemize}
Thus, \texttt{strong\_hypothesis\_rho} reflects the best competitive ratio compatible with the current potential form.

\paragraph{Bellman-Based Guidance Metrics.}

For diagnostics, we compute the Bellman potential (via Bellman–Ford) in the normalized graph with edge weights
\[
(k+1)\,\Delta{\mathrm{OPT}} - \nabla.
\]
Although such numeric potentials do not yield symbolic proofs, they provide valuable guidance for constructing structured potentials.

\paragraph{\texttt{bellman\_\{node/edge\}\_\{r2,corr\}}.}  
Coefficient of determination ($R^2$) or correlation between:
\begin{itemize}
    \item proposed potential values $\Phi$, and
    \item Bellman–Ford potential values,
\end{itemize}
either across nodes (\texttt{node}) or across edges (\texttt{edge}, comparing $\Phi(w_v)-\Phi(w_u)$).  

High correlation indicates that the candidate potential resembles the Bellman–Ford solution, though this does not guarantee validity.

\subsection{Search Framework}

Our \emph{search} formulation introduces three interfaces: \texttt{Potential}, \texttt{PotentialFamily}, and \texttt{SearchEvaluator}. The \texttt{Potential} interface remains unchanged from the non-search setting. The \texttt{PotentialFamily} interface is responsible for driving optimization over the space of candidate potentials, typically through an ask/tell interaction loop. Finally, \texttt{SearchEvaluator} provides a \emph{proxy} evaluation signal for proposed candidates. In the simplest implementation, it computes the potential on all nodes and returns the total number of violated inequalities; however, this can become prohibitively expensive on larger instances.

The search process proceeds as follows. Within a fixed time budget, \texttt{PotentialFamily} is periodically \emph{asked} to propose candidate potentials for evaluation. Each proposed candidate is sent to a worker pool and scored by \texttt{SearchEvaluator}. The resulting feedback is then \emph{told} back to \texttt{PotentialFamily}, which may use it to update its internal search state and propose improved candidates in subsequent rounds.

At the end of the allotted time, \texttt{PotentialFamily} returns a set of final candidates, which are then evaluated on the full benchmark instances using the complete scoring procedure. The score of the search method is defined by the best candidate found during this process.

We provide default implementations for all three components: a naive \texttt{Potential}, a canonical \texttt{Potential}, a naive \texttt{SearchEvaluator} that exhaustively computes all violations, and a naive \texttt{PotentialFamily} that simply returns the supplied hyperparameters without adaptation.

We also experimented with a less restrictive evaluation protocol in which participants are only required to define a \texttt{Potential} and output hyperparameters after a fixed timeout. Under this protocol, submissions are executed as standalone programs with a specified time budget and CPU allocation.

\section{Methods Descriptions and Used Hyperparameters}
\label{sec: Method Description and Used Hyperparameters}

The hyperparameters for \textsc{ShinkaEvolve} and \textsc{LoongFlow} were selected based on the example configurations provided in their public repositories. In our preliminary ablations, moderate changes to these parameters did not lead to substantial differences in performance.

\paragraph{Best-of-$N$.}
In each Best-of-$N$ iteration, we allow up to two sampling attempts to obtain a parseable response from the LLM. We use temperature $\tau = 0.8$ and cap responses at \texttt{max\_tokens}=16{,}000.

\paragraph{\textsc{ShinkaEvolve}.}
\textsc{ShinkaEvolve} emphasizes \emph{sample efficiency} through mechanisms that balance exploration and exploitation (e.g., adaptive parent sampling), encourage diversity (e.g., novelty-based rejection sampling of code proposals), and adaptively allocate generation effort across models, while iteratively improving program quality.

In the $k=3$ testbed experiments, we use the same model for coding, novelty, and meta tasks. We employ \emph{weighted} parent sampling with \texttt{parent\_selection\_ratio}=10. Each prompt includes 4 archive inspirations and 2 top-$k$ inspirations. The archive size is set to 40, and the elite selection ratio is 0.3. Migration is performed every 10 iterations with migration rate 0.1. We use 4 islands. Generations are capped at 80{,}000 tokens, and patch types are sampled with probabilities 0.6, 0.3, and 0.1 for \texttt{diff}, \texttt{full}, and \texttt{cross}, respectively.

\paragraph{\textsc{LoongFlow}.}
\textsc{LoongFlow} replaces ``blind'' mutation with a more deliberative \emph{Plan--Execute--Summarize} loop, and combines this with an evolutionary memory system and diversity-preserving selection to support longer-horizon improvement and reduce premature convergence.

We use the same model for the planner, executor, and summarizer. The planner is allocated 8 rounds; the executor uses 4 rounds with branching factor 4; and the summarizer is allowed up to 4 ReAct steps. The database uses 3 islands with population size 40, and the parent-sampling parameter is set to 2.

\begin{figure}[!t]
    \centering
    \includegraphics[width=\linewidth]{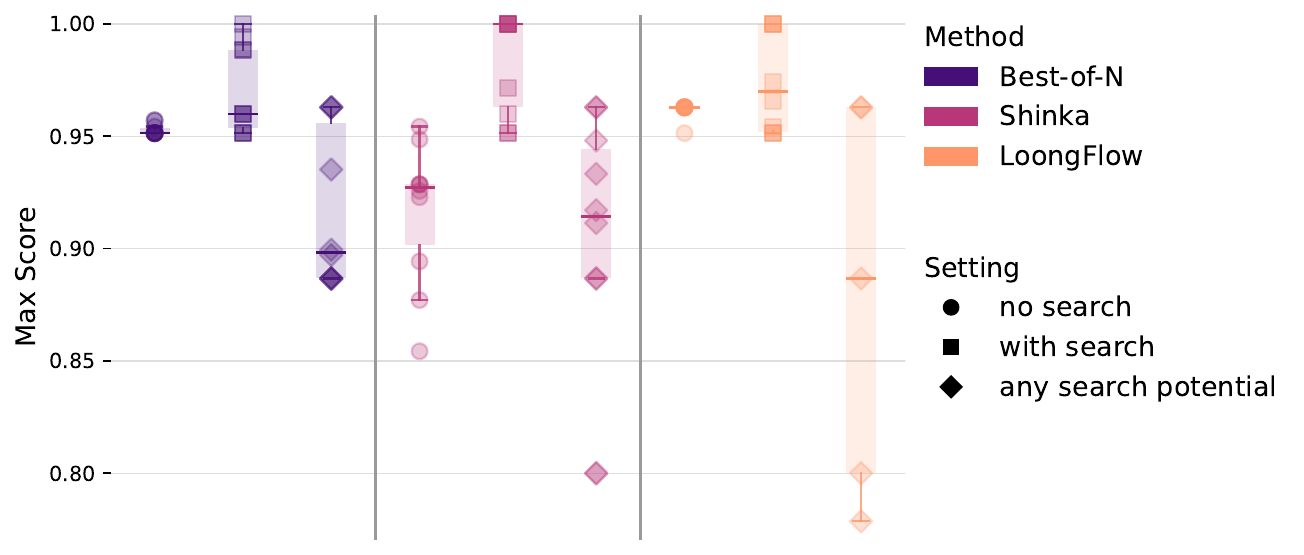}
    \caption{10 independent runs of \textsc{ShinkaEvolve}, Best-of-$N$, \textsc{LoongFlow} per task design: (i) no search, just suggest a canonical potential, (ii) with search over canonical potentials, (iii) with search without canonical potential hint. All methods had the same starting conditions per task design and had 100 evaluations cap.}
    \label{fig: Comparing Search, With Search, Any Search Potential}
\end{figure}

\section{$k=3$ testbed experiments}
\label{sec: k=3 testbed experiments}

To study the effect of task design, we consider three settings: 
(i) \emph{no search} with a canonical potential hint, 
(ii) \emph{search} with a canonical hint, and 
(iii) \emph{search} without a canonical hint. 
\autoref{fig: Comparing Search, With Search, Any Search Potential} shows box plots of the best scores achieved across runs and \autoref{tab: Support Summary comparing Search, With Search, Any Search Potential} shows summary statistics of these runs.

\begin{table}[!t]
    \centering
    \caption{Support summary for the \autoref{fig: Comparing Search, With Search, Any Search Potential}. ``-search'' means no search with canonical hint setting, ``+search'' means with search setting with canonical hint and ``any'' means with search and without canonical hint setting. The evaluations were done on one $k = 3, m = 6$ without $k$-taxi augmentations instance.}
    \vspace{3pt}
    \label{tab: Support Summary comparing Search, With Search, Any Search Potential}
    \scriptsize
\begin{tabularx}{\textwidth}{>{\raggedright\arraybackslash}X*{9}{c}}
\toprule
Metric & \multicolumn{3}{c}{Best-of-N} & \multicolumn{3}{c}{Shinka} & \multicolumn{3}{c}{LoongFlow} \\
\cmidrule(lr){2-4}\cmidrule(lr){5-7}\cmidrule(lr){8-10}
 & -search & +search & any & -search & +search & any & -search & +search & any \\
\midrule
Perfect & 0 & 1 & 0 & 0 & \textbf{6} & 0 & 0 & 4 & 0 \\
Mean Final & 0.953 & 0.971 & 0.948 & 0.916 & \textbf{0.983} & 0.901 & 0.962 & 0.975 & 0.866 \\
Median Final & 0.951 & 0.960 & 0.963 & 0.927 & \textbf{1.000} & 0.914 & 0.963 & 0.970 & 0.845 \\
Min Final & \textbf{0.951} & \textbf{0.951} & 0.897 & 0.854 & \textbf{0.951} & 0.800 & \textbf{0.951} & \textbf{0.951} & 0.779 \\
Max Final & 0.957 & \textbf{1.000} & 0.963 & 0.954 & \textbf{1.000} & 0.963 & 0.963 & \textbf{1.000} & 0.963 \\
Mean Attempt & 0.796 & 0.764 & 0.778 & 0.790 & \textbf{0.916} & 0.785 & 0.424 & 0.621 & 0.438 \\
Median Attempt & 0.786 & 0.731 & 0.800 & 0.789 & \textbf{0.951} & 0.777 & 0.649 & 0.771 & 0.649 \\
\bottomrule
\end{tabularx}

\end{table}

The \emph{search with canonical hint} setting consistently outperforms the alternatives across all methods. In this regime, \textsc{ShinkaEvolve} slightly outperforms \textsc{LoongFlow}, while Best-of-$N$ never achieves a perfect score, although it comes close. In the \emph{no search} setting, Best-of-$N$ outperforms its \textsc{ShinkaEvolve} counterpart and performs only slightly worse than \textsc{LoongFlow}, suggesting that without search the differences between agentic scaffolds become marginal. Finally, the \emph{search without canonical hint} setting exhibits the highest variance; however, the best discovered potentials are comparable to those obtained in the no-search setting, indicating difficulty in navigating the unrestricted hypothesis space.

We hypothesize that this behavior is driven by the structure of good potentials in the $k$-server domain. High-quality solutions tend to exhibit strong symmetry and global structure, making them difficult to discover via incremental, locally guided improvements. This distinction between tasks that favor global construction versus iterative refinement has been noted in \textcite{novikov2025alphaevolvecodingagentscientific}. In particular, when the target is a symmetric mathematical object, large-scale sampling methods such as FunSearch \parencite{Romera-Paredes2024-cn} can be effective due to their ability to explore many candidates, while iterative methods struggle to extract useful signal from the evaluator. Conversely, approaches such as AlphaEvolve \parencite{novikov2025alphaevolvecodingagentscientific} and other LLM-driven evolutionary systems are better suited for problems where evaluation is expensive, heuristic guidance is informative, and incremental improvements are possible.

In our setting, the search paradigm remains advantageous despite these challenges, as it enables scaling through additional evaluations and avoids wasting LLM calls on trivial hyperparameter variations, which can instead be handled within the generated programs. This allows frontier models to focus their capacity on proposing structurally meaningful candidates, making search-based methods a natural fit for this benchmark.

\section{$k = 4$}
\label{sec: k = 4 scaling challenges}

\paragraph{Token Consumption Analysis.}

To better understand the cost side of proxy-based search, we inspected a matched pair of \textsc{Shinka} runs. The two runs used the same task family, the same prompt construction logic, and the same search scaffold (same numbers of inspirations, archive settings, mutation modes, and proxy-evaluation interface), but differed in the model pool: one run used only \texttt{gpt-oss-120b}, while the other used a frontier mixture of \texttt{gpt-5.2}, \texttt{gpt-5-mini}, and \texttt{gpt-oss-120b} (with \texttt{gpt-5.2} also handling meta-summaries). Importantly, the prompt only specified the generic $k$-server potential-search task on the four $k=4$ circle / circle-taxi instances.

Despite this matched setup, the frontier run consumed far more tokens. The frontier run used $28.8$M total tokens across $399$ model-generated iterations, compared with $11.5$M tokens across $395$ iterations for the OSS-only run, a factor of about $2.5\times$. The frontier run did achieve a better best score, improving from $0.4819$ (OSS-only) to $0.5937$, but this gain came at substantially worse token efficiency. 

The main reason was not a higher retry rate or a different interface contract. In fact, retry statistics were comparable, and in some cases slightly worse for the OSS-only run. Instead, the frontier models quickly produced much larger candidate programs, and \textsc{Shinka}'s prompt builder recursively pastes the current program together with several previously discovered programs back into later prompts. As a result, average candidate code size nearly doubled (about $28.4$k versus $14.2$k characters), and average prompt size more than doubled (about $203$k versus $90.5$k characters). This creates a self-amplifying long-context loop: larger code leads to larger prompts, which in turn encourages still larger edits.

\paragraph{Search Summary.}

We carried out a broad $k=4$ search campaign across multiple regimes, ranging from early canonical-structured searches to later structurally hinted and interface-ablated setups. In total, the saved outputs contain $44$ launches, $139$ individual runs, and approximately $20,429$ LLM-generated evaluated candidates ($9,048$ stored \textsc{Shinka} programs and $11,381$ \textsc{LoongFlow} evaluation artifacts). The largest share of this effort was spent in the structurally hinted family, which alone accounts for roughly $10,577$ evaluated candidates. Despite this scale, no experiment produced a complete $k=4$ solution. The best saved result reached combined score 0.9941 with violation vector [0, 96, 9427, 663], so even the strongest candidate still incurred about 9.4k violations on the hardest taxi-augmented $m=6$ instance.

The remaining regimes support the same conclusion. In the early canonical / fixed-structure family, the best saved run reached combined score 0.9844 with violation vector [0, 368, 10524, 2020], while one intermediate candidate briefly reduced the total violation count to 7249 with vector [0, 32, 515, 6702]. Later interface-ablation runs also improved incumbents but did not change the qualitative picture: the best PotentialFamily (without SearchEvaluator and fixed CanonicalPotential) run reached 0.9092 with [114, 11688, 65128, 3294], and the best run with simplified evaluation interface follow-up reached 0.8579 with [210, 15752, 38282, 7762].

Most of these experiments were run with `gpt-oss` models rather than frontier models. The reason is the cost pathology discussed in the previous paragraph: frontier models tend to
generate substantially larger candidate programs, and both \textsc{Shinka} and \textsc{\textsc{LoongFlow}} construct new prompts by pasting in multiple complete prior solutions. As solution size grows, prompt size grows with it, so token usage increases superlinearly over the course of a run. In practice, this made frontier-model sweeps much more expensive without changing the qualitative $k=4$ outcome, and therefore most of the large-scale search effort reported here relied on OSS models.

\subsection{Structural Hints}
\label{sec: Structural Hints}

\paragraph{Unifying Hint.}

In addition to the standard task definition, canonical-potential description, and interface prompt, we provided the following hint:

\begin{tcolorbox}[
  title=Unifying Hint,
  colback=gray!5,
  colframe=black!70,
  boxrule=0.8pt,
  arc=2mm,
  left=2mm,
  right=2mm,
  top=1mm,
  bottom=1mm,
  fonttitle=\bfseries,
  breakable
]

The particularly strong `index\_matrix` is the following:

[[1, 2, 3, 4],
 [-1, 2, 3, 4],
 [-2, -2, 3, 4],
 [-3, -3, -3, 4],
 [-4, -4, -4, -4]]

Search for matrices around this matrix and coefficients that are suited for this matrix. 
\end{tcolorbox}

One suggestion produced by ShinkaEvolve was to consider $n=5$ instead of $n=4$ while keeping the same \texttt{index\_matrix}. Intuitively, this introduces an additional auxiliary point into the minimization, which interacts with the existing points only through the distance coefficients. More details are given in \autoref{sec: Potential, Codex with Human-in-the-Loop}.

We also exhaustively checked all coefficient vectors with support in $\{-1,0,1\}$ for the given matrix, and confirmed that $14$ was indeed the minimum achievable value in this restricted search space.

\paragraph{Symmetricity Hint.}

\begin{tcolorbox}[
  title=Symmetric Hint,
  colback=gray!5,
  colframe=black!70,
  boxrule=0.8pt,
  arc=2mm,
  left=2mm,
  right=2mm,
  top=1mm,
  bottom=1mm,
  fonttitle=\bfseries,
  breakable
]

Try satisfying the following requirements for your search over matrices and coefficients:

- Fix n = 5
- One row should be all -1's
- Every other row should contain 1
- Every other index should be balanced: appear the same number of times with positive and negative signs

Do the search around these matrices and suiting coefficients.
\end{tcolorbox}

Under this hint, we additionally explored several matrices satisfying these structural requirements, including:
\[
    \begin{pmatrix}
        -5 & -5 & -5 & -5 \\
        5 & -1 & -1 & -1 \\
        5 & 1 & 2 & -3 \\
        5 & 1 & 3 & -4 \\
        5 & 1 & 4 & -2
    \end{pmatrix}
    \qquad \text{and} \qquad
    \begin{pmatrix}
        -5 & -5 & -5 & -5 \\
        5 & -1 & -2 & -2 \\
        5 & 1 & 3 & 4 \\
        5 & 2 & -4 & -4 \\
        5 & 2 & 4 & -3
    \end{pmatrix}.
\]

Here the index $5$ is simply a relabeling of the distinguished point previously denoted by $1$. Among these candidates, the second matrix yielded better empirical performance. 

\subsection{Codex-found Solution}
\label{sec: Codex-found Solution}

In this phase, the search was organized around the hypothesis that the main structural difficulty lay in choosing the right canonical $n=7$ index matrix, after which the remaining improvement could be obtained by coefficient optimization alone. The search therefore first stabilized on a highly symmetric five-row $n=7$ ansatz with one distinguished anchor variable and a cyclic arrangement of the remaining variables, and then froze this matrix while focusing entirely on the pair-coefficient vector. Early coefficient searches were deliberately conservative: they explored sparse local perturbations with coefficients essentially restricted to small $\{-1,0,1\}$-type moves, which was already enough to reduce the taxi instance to 6 violations and then 5 violations. The next step was to enlarge the coefficient search neighborhood by allowing larger signed updates and a wider integer range, so that the search was no longer limited to single-step local corrections; this richer coefficient dynamics enabled the jump from the first near-solutions to the first 3-violation hit.

From an engineering point of view, the \texttt{Potential} and \texttt{SearchEvaluator} were designed to make very short search iterations informative. For a fixed $n=7$ matrix, the \texttt{Potential} implementation precomputed the configuration-index tables induced by the rows, as well as all pair-distance values over the full assignment space, so that evaluating a candidate potential reduced to summing a few preindexed work-function entries, subtracting a precomputed linear penalty, and taking a minimum. The \texttt{SearchEvaluator} then exploited this by evaluating candidates on sampled edge sets rather than full scans during search. It memoized node potentials within each evaluation so that each node was scored at most once, mixed random edge samples with a cache of previously discovered hard edges, and stopped early after a prescribed number of violations. In this way, each failed candidate not only produced a score, but also exposed a small set of particularly informative violated edges; these were recycled into subsequent evaluations, so the search increasingly concentrated on the difficult parts of the taxi instance instead of repeatedly spending effort on easy edges.

The \texttt{PotentialFamily} acted as a lightweight staged local-search controller on top of this evaluator. It maintained a queue of candidate coefficient vectors on the fixed symmetric matrix, separated into a cheap “quick” stage and a more reliable “confirm” stage. Candidates were seeded from a small family of structured sparse coefficient patterns and then expanded by repeated local mutations, initially with single-step sign changes and later with larger jumps and larger allowable magnitudes. Only coefficients not involving the distinguished auxiliary variable were actively mutated, which reduced the search dimension and preserved the symmetry assumptions built into the matrix ansatz. Candidates that performed well under quick sampled evaluation were promoted to confirm runs with larger edge budgets, while hard violated edges collected from these runs were added back into the evaluator’s cache. Thus the overall method was not a black-box optimizer, but a deliberately engineered loop: first identify a strong symmetric $n=7$ matrix, then run an increasingly expressive coefficient search around it, with cached hard-edge feedback and staged evaluation guiding the progression from 6 to 5 and finally to 3 violations.

\end{document}